\documentclass{llncs}
\usepackage{tikz,tkz-graph,amsmath,amsfonts,amssymb,graphics,color,comment,graphicx,fullpage}
\usetikzlibrary{arrows.meta}

\newcommand{\NP}{{\sf NP}}

\newtheorem{observation}{Observation}

\begin{document}

\title{Classifying $k$-Edge Colouring for $H$-free Graphs\thanks{The second author was supported by  the Research Council of Norway via the project CLASSIS. The third author was supported by the Leverhulme Trust (RPG-2016-258).}}
\author{
Esther Galby\inst{1}
\and 
Paloma T. Lima\inst{2}
\and
Dani{\"e}l Paulusma\inst{3}
\and
Bernard Ries\inst{1}
}

\institute{
Department of Informatics, University of Fribourg, Switzerland,  \texttt{\{esther.galby,bernard.ries\}@unifr.ch}
\and
Department of Informatics, University of Bergen, Norway, \texttt{paloma.lima@uib.no}
\and
Department of Computer Science, Durham University, UK, \texttt{daniel.paulusma@durham.ac.uk}
}

\maketitle

\begin{abstract}
A graph is $H$-free if it does not contain an induced subgraph isomorphic to~$H$.
For every integer $k$ and  every graph~$H$, we determine the computational complexity of $k$-{\sc Edge Colouring} for $H$-free graphs.
\end{abstract}

\section{Introduction}\label{s-intro}

A graph $G=(V,E)$ is {\it $k$-edge colourable} for some integer~$k$ if there exists a mapping $c:E\to \{1,\ldots,k\}$ such that
$c(e)\neq c(f)$ for any two edges $e$ and $f$ of $G$ that have a common end-vertex. The {\it chromatic index} of $G$ is the smallest integer~$k$ such that $G$ is $k$-edge colourable. Vizing proved the following classical result.

\begin{theorem}[\cite{Vi64}]\label{t-vizing}
The chromatic index of a graph $G$ with maximum degree~$\Delta$ is either $\Delta$ or $\Delta+1$.
\end{theorem}

The {\sc Edge Colouring} problem is to decide if a given graph $G$ is $k$-edge colourable for some given integer~$k$.
Note that $(G,k)$ is a yes-instance if $G$ has maximum degree at most $k-1$ by Theorem~\ref{t-vizing} and that
$(G,k)$ is a no-instance if $G$ has maximum degree at least $k+1$.
If $k$ is {\it fixed}, that is, $k$ is not part of the input, then we denote the problem by $k$-{\sc Edge Colouring}.
It is trivial to solve this problem for $k=2$.
However, the problem is \NP-complete if $k\geq 3$, as shown by Holyer for $k=3$ and by Leven and Galil for $k\geq 4$.

\begin{theorem}[\cite{Ho81,LG83}]\label{t-hard}
For $k\geq 3$, {\sc $k$-Edge Colouring} is \NP-complete even for $k$-regular graphs.
\end{theorem}

Due to the above hardness results we may wish to restrict the input to some special graph class. 
A natural property of a graph class is to be closed under vertex deletion. Such graph classes are called {\it hereditary} and they form the focus of our paper. To give an example, bipartite graphs form a hereditary graph class, and it is well-known that they have chromatic index $\Delta$. Hence, {\sc Edge Colouring} is polynomial-time solvable for bipartite graphs, which are perfect and triangle-free. In contrast, Cai and Ellis~\cite{CE91} proved that for every $k\geq 3$, $k$-{\sc Edge Colouring} is \NP-complete for $k$-regular comparability graphs, which are also perfect. 
They also proved the following two results, the first one of which shows that {\sc Edge Colouring} is \NP-complete for triangle-free graphs (the graph $C_s$ denotes the cycle on $s$ vertices).

\begin{theorem}[\cite{CE91}]\label{t-ce1}
Let $k\geq 3$ and $s\geq 3$. Then $k$-{\sc Edge Colouring} is \NP-complete for $k$-regular $C_s$-free graphs.
\end{theorem}

\begin{theorem}[\cite{CE91}]\label{t-ce2}
Let $k\geq 3$ be an odd integer. Then $k$-{\sc Edge Colouring} is \NP-complete for $k$-regular line graphs of bipartite graphs.
\end{theorem}

It is also known that {\sc Edge Colouring} is polynomial-time solvable for chordless graphs~\cite{MFT13},
series-parallel graphs~\cite{Jo85}, split-indifference graphs~\cite{OMS98} and for graphs of treewidth at most $k$ for any constant~$k$~\cite{Bo90}. 

It is not difficult to see that a graph class~${\cal G}$ is hereditary if and only if it can be characterized  by a set ${\cal F}_{\cal G}$ of forbidden induced subgraphs (see, for example,~\cite{KL15}).
Malyshev determined the complexity of  $3$-{\sc Edge Colouring} for every hereditary graph class ${\cal G}$, for which ${\cal F}_{\cal G}$ consists of graphs that each have at most five vertices, except perhaps two graphs that may contain six vertices~\cite{Ma14}. 
Malyshev performed a similar complexity study for {\sc Edge Colouring} for graph classes defined by a family of forbidden (but not necessarily induced) graphs with at most seven vertices and at most six edges~\cite{Ma17}.

 We focus on the case where ${\cal F}_{\cal G}$ consists of a single graph~$H$. A graph $G$ is {\em $H$-free} if $G$ does not contain an induced subgraph isomorphic to~$H$.
We obtain the following dichotomy for $H$-free graphs.

\begin{theorem}\label{t-main}
Let $k\geq 3$ be an integer and $H$ be a graph. If $H$ is a linear forest, then $k$-{\sc Edge Colouring} is polynomial-time solvable for  $H$-free graphs. Otherwise $k$-{\sc Edge Colouring} is \NP-complete even for $k$-regular $H$-free graphs.
\end{theorem}
We obtain Theorem~\ref{t-main} by combining Theorems~\ref{t-ce1} and~\ref{t-ce2} with two new results. 
In particular, we will prove a hardness result for $k$-regular claw-free graphs for even integers~$k$ (as Theorem~\ref{t-ce2} is only valid when $k$ is odd).

\section{Preliminaries}\label{s-pre}

The graphs $C_n$, $P_n$ and $K_n$ denote the path, cycle and complete graph on $n$ vertices, respectively.
A set $I$ is an {\it independent set} of a graph~$G$ if all vertices of $I$ are pairwise nonadjacent in~$G$.
A graph $G$ is {\it bipartite} if its vertex set can be partitioned into two independent sets~$A$ and $B$. If there exists an edge between every vertex of $A$ and every vertex of $B$, then $G$ is {\it complete bipartite}.
The {\it claw} $K_{1,3}$ is the complete bipartite graph with $|A|=1$ and $|B|=3$.

Let $G_1$ and $G_2$ be two vertex-disjoint graphs.
The {\it join} operation~$\times$  adds an edge between every vertex of $G_1$ and every vertex of $G_2$.
The {\it disjoint union}  operation~$+$ merges $G_1$ and $G_2$ into one graph without adding any new edges, that is,
$G_1+G_2=(V(G_1)\cup V(G_2), E(G_1)\cup E(G_2))$. We write $rG$ to denote the disjoint union of~$r$ copies of a graph~$G$.

A {\it forest} is a graph with no cycles. 
A {\it linear forest} is a forest of maximum degree at most~2, or equivalently, a disjoint union of one or more paths. 
A graph $G$ is a {\it cograph} if $G$ can be generated from $K_1$ by a sequence of join and disjoint union operations.
A graph is a cograph if and only if it is $P_4$-free (see, for example,~\cite{BLS99}).
The following well-known lemma follows from this equivalence and the definition of a cograph.

\begin{lemma}\label{l-p4}
Every connected $P_4$-free graph on at least two vertices has a spanning complete bipartite subgraph.
\end{lemma}

Let $G=(V,E)$ be a graph. For a subset $S\subseteq V$, the graph $G[S]=(S,\{uv\in E\; |\; u,v\in S\})$ denotes the subgraph of $G$ {\it induced} by $S$. We say that $S$ is {\it connected} if $G[S]$ is connected.
Recall that a graph $G$ is {\it $H$-free} for some graph~$H$ if $G$ does not contain $H$ as an induced subgraph.
 A subset $D\subset V(G)$ is {\it dominating} if every vertex of $V(G)\setminus D$ is adjacent to least one vertex of $D$.
We will need the following result of Camby and Schaudt.
  
\begin{theorem}[\cite{CS16}]\label{t-cs}
Let $t\geq 4$ and $G$ be a connected $P_t$-free graph. Let $X$ be any minimum connected dominating set of~$G$. Then 
$G[X]$ is either $P_{t-2}$-free or isomorphic to $P_{t-2}$. 
\end{theorem}

Let $G=(V,E)$ be some graph. The {\it degree} of a vertex~$u\in V$ is equal to the size of its neighbourhood $N(u)=\{v\; |\; uv\in E\}$. The graph~$G$ is {\it $r$-regular} if every vertex of $G$ has degree~$r$.
The {\it line graph} of $G$ is the graph~$L(G)$, which has vertex set $E$ and an edge between two distinct vertices $e$ and $f$ if and only if $e$ and~$f$ have a common end-vertex in $G$.

\section{The Proof of Theorem~\ref{t-main}}

To prove our dichtomy,
we first consider the case where the forbidden induced subgraph $H$ is a claw. As line graphs are claw-free, we only need to deal with the case where the number of colours~$k$ is even due to Theorem~\ref{t-ce2}.
For proving this case we need another result of Cai and Ellis, which we will use as a lemma.
Let $c$ be a $k$-edge colouring of a graph $G=(V,E)$. Then a vertex~$u\in V$ {\it misses} colour~$i$ if none of the edges incident to $u$ is coloured~$i$.

\begin{lemma} [\cite{CE91}]\label{Kneven}
For even $k\geq 2$, the complete graph $K_k$ has a $k$-edge colouring with the property that~$V(K_k)$ can be partitioned into sets 
$\{u_i,u_i'\}$ $(1\leq i\leq \frac{k}{2})$, such that for $i=1,\ldots,\frac{k}{2}$, vertices $u_i$ and $u_i'$ miss the same colour, which is not missed by any of the other vertices. 
\end{lemma}

We use Lemma~\ref{Kneven} to prove the following result, which solves the case where $k$ is even and $H=K_{1,3}$. 

\begin{lemma}\label{l-claw}
Let $k\geq 4$ be an even integer. Then $k$-{\sc Edge Colouring} is \NP-complete for $k$-regular claw-free graphs.
\end{lemma}

\begin{proof}
Recall that $k$-{\sc Edge Colouring} for $k$-regular graphs is \NP-complete for every integer $k\geq 4$ due to Theorem~\ref{t-hard}. Consider an instance $(G,k)$ of $k$-{\sc Edge Colouring}, where $G$ is $k$-regular for some even integer $k=2\ell\geq 4$. From~$G$ we construct a graph~$G'$ as follows. First we replace every vertex $v$ in $G$ by the gadget~$H(v)$ shown in Figure \ref{Fig:gadgetH}.  Next we connect the different gadgets in the following way.
Every gadget $H(v)$ has exactly $k$ pendant edges, which are incident with vertices $v_{1},\ldots,v_{\ell},v_{\ell+1},\ldots,v_{2\ell}$, respectively. As $G$ is $k$-regular, every vertex has $k$ neighbours in $G$. 
Hence,  we can identify each edge $uv$ of $G$ with a unique edge $u_hv_i$ in $G'$, which is a pendant edge of both $H(u)$ and $H(v)$. It is readily seen that $G'$ is $k$-regular and claw-free.

\tikzset{
  circ/.style = {circle,draw,fill,inner sep=1pt},
  invisible/.style = {circle,draw=none,inner sep=0pt,font=\tiny}
}

\begin{figure}
\centering
\begin{tikzpicture}[scale=2,node distance=.7cm]
\node[circ,label=below:{\footnotesize $v$}] (v) at (-2.5,0) {};
\node[circ,above right of =v] (u1) {};
\node[circ,above of=v] (u2) {};
\node[circ,above left of=v] (u3) {};
\node[circ,below left of=v] (u4) {};
\node[circ,below right of=v] (u5) {};

\draw[thick,dotted,domain=168:192] plot ({0.3*cos(\x)-2.5},{0.3*sin(\x)});
\draw[thick,dotted,domain=348:372] plot ({0.3*cos(\x)-2.5},{0.3*sin(\x)});

\draw[thick] (v) -- (u1)
(v) -- (u2)
(v) -- (u3)
(v) -- (u4)
(v) -- (u5);

\draw[-Implies,line width=.6pt,double distance=2pt] (-1.8,0) -- (-1.4,0);

\node[circ,label=below:{\footnotesize $v'_1$}] (v1) at (.09,.99) {};
\node[circ,label=below:{\footnotesize $v'_2$}] (v2) at (.34,.94) {};
\node[circ,label=above:{\footnotesize $v'_{\ell}$}] (vl) at (.09,-.99) {};
\node[circ,label=above:{\footnotesize $v_{1}$}] (vl1) at (-.17,-.98) {};
\node[circ,label=below:{\footnotesize $v_{\ell}$}] (v2l) at (-.17,.98) {};
\node[circ,label={[label distance=.05cm]178: \footnotesize $v_{\ell -1}$}] (v2l1) at (-.42,.9) {};
\node[circ,label={[label distance=.05cm]182: \footnotesize $v_{2}$}] (vl2) at (-.42,-.9) {};
\node[draw=none] (K1) at (0,0) {$K_1(v)$};

\node[circ] (p2l) at (-.21,1.18) {};
\node[circ] (pl1) at (-.21,-1.18) {};
\node[circ] (p2l1) at (-.5,1.09) {};
\node[circ] (pl2) at (-.5,-1.09) {};

\draw[-,thick] (v2l) -- (p2l)
(vl1) -- (pl1)
(v2l1) -- (p2l1)
(vl2) -- (pl2);

\draw[thick,domain=60:125] plot ({cos(\x)}, {sin(\x)});
\draw[thick,domain=235:290] plot ({cos(\x)}, {sin(\x)});
\draw[thick,dotted,domain=100:260] plot ({cos(\x)}, {sin(\x)});
\draw[thick,dotted,domain=280:430] plot ({cos(\x)}, {sin(\x)});

\node[circ,label=below:{\footnotesize $w$}] (w) at (2,0) {};

\node[circ,label=below:{\footnotesize $v_{\ell +1}$}] (v'2l) at (4.17,.98) {};
\node[circ,label={[label distance=.05cm]2: \footnotesize $v_{\ell +2}$}] (v'2l1) at (4.42,.9) {};
\node[circ,label=above:{\footnotesize $v_{2\ell}$}] (v'l1) at (4.17,-.98) {};
\node[circ,label=below:{\footnotesize $v'_{\ell+1}$}] (v'l) at (3.91,-.99) {};
\node[circ,label=below:{\footnotesize $v'_{2\ell}$}] (v'1) at (3.91,.99) {};
\node[circ,label=below:{\footnotesize $v'_{2\ell -1}$}] (v'2) at (3.66,.94) {};
\node[circ,label={[label distance=.05cm]-3: \footnotesize $v_{2\ell -1}$}] (v'l2) at (4.42,-.9) {};
\node[draw=none] (K2) at (4,0) {$K_2(v)$};

\node[circ] (p'2l) at (4.21,1.18) {};
\node[circ] (p'l1) at (4.21,-1.19) {};
\node[circ] (p'2l1) at (4.5,1.09) {};
\node[circ] (p'l2) at (4.5,-1.09) {};

\draw[-,thick] (v'2l) -- (p'2l)
(v'l1) -- (p'l1)
(v'2l1) -- (p'2l1)
(v'l2) -- (p'l2);

\draw[thick,domain=55:120] plot ({4+cos(\x)}, {sin(\x)});
\draw[thick,domain=250:305] plot ({4+cos(\x)}, {sin(\x)});
\draw[thick,dotted,domain=100:260] plot ({4+cos(\x)}, {sin(\x)});
\draw[thick,dotted,domain=280:430] plot ({4+cos(\x)}, {sin(\x)});

\draw[thick] (w) ..controls (.75,.9).. (v2);
\draw[thick] (w) ..controls (.6,1.1).. (v1);
\draw[thick] (w) ..controls (.6,-1.1).. (vl);
\draw[thick] (w) ..controls (3.25,.9).. (v'2);
\draw[thick] (w) ..controls (3.4,1.1).. (v'1);
\draw[thick] (w) ..controls (3.4,-1.1).. (v'l);
\end{tikzpicture}
\caption{The gadget $H(v)$ where $K_i(v)$ is a complete graph of size $2\ell$ for $i=1,2$. 
Note that edges inside $K_1(v)$ and $K_2(v)$ are not drawn.}
\label{Fig:gadgetH}
\end{figure}
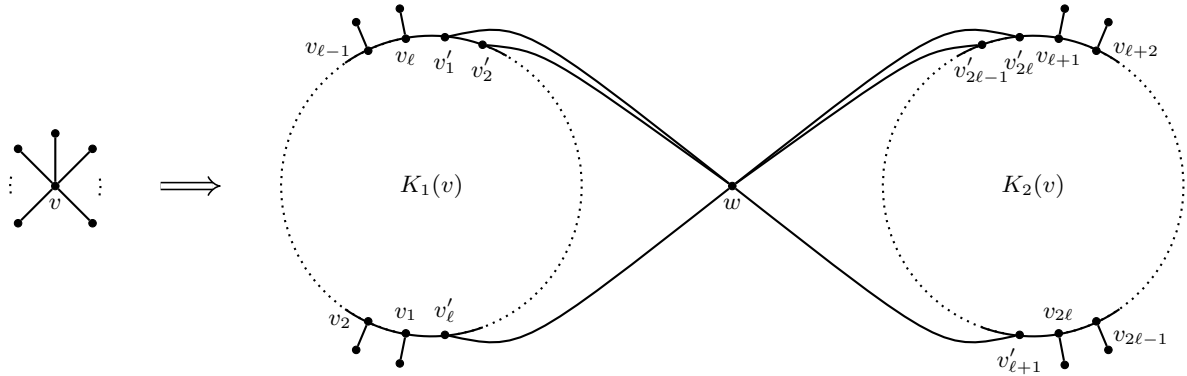

First suppose that $G$ is $k$-edge colourable. Let $c$ be a $k$-edge colouring of~$G$. Consider a vertex $v\in V(G)$. For every neighbour~$u$ of $v$ in $G$, we colour the pendant edge in $H(v)$ corresponding to the edge $uv$ with colour $c(uv)$. As $c$ assigned different colours to the edges incident to $v$, the $2\ell$ pendant edges of $H(v)$ will receive pairwise distinct colours, which 
we denote by $x_1,\ldots,x_\ell, y_1,\ldots, y_\ell$.  By Lemma~\ref{Kneven}, we can colour the edges of $K_1(v)$ in such a way that  for $i=1,\ldots,\ell$, $v_i$ and $v'_i$ miss colour $x_i$. For $i=1,\ldots,\ell$, we can therefore assign colour~$x_i$ to edge~$v'_iw$. Similarly, we may assume that  for $i=1,\ldots,\ell$,  $v_{\ell +i}$ and $v_{\ell+i}'$ miss colour~$y_i$. For $i=1,\ldots,\ell$, we can therefore assign colour~$y_i$ to edge~$v_{\ell +i}'w$. Recall that the colours $x_1,\ldots,x_\ell$, $y_1,\ldots,y_\ell$ are all different. Hence, doing this procedure for each vertex of~$G$ yields a $k$-edge colouring~$c'$ of~$G'$.

Now suppose that $G'$ is $k$-edge colourable. Let $c'$ be a $k$-edge colouring of~$G'$. Consider some $v\in V(G)$. Denote the pendant edges of $H(v)$ by $e_i$ for 
$i=1,\ldots,2\ell$, where $e_i$ is incident to $v_i$ (and to some vertex $u_h$ in a gadget $H(u)$ for each neighbour~$u$ of $v$ in $G$). Suppose that $c'$ gave colour~$x$ to an edge $wv_i'$ for some $1\leq i\leq \ell$, say to $wv_1'$, but not to any edge $e_i$ for $i=1,\ldots, \ell$. Note that $wv_2',\ldots,wv'_\ell$ cannot be coloured~$x$. As every vertex of $G'$ has degree $k=2\ell$, every $v_i$ with $1\leq i\leq \ell$ and every $v_j'$ with $2 \leq j \leq \ell$ is incident to some edge coloured $x$. As~$x$ is neither the colour of $e_{1},\ldots,e_{\ell}$ nor the colour of $wv_2',\ldots,wv'_\ell$, the complete graph $K_1(v)-v'_1$ contains a perfect matching all of whose edges have colour~$x$. However, $K_1(v)-v'_1$ has odd size~$2\ell-1$.
Hence, this is not possible. We conclude that each of the (pairwise distinct) colours of $wv'_1,\ldots,wv'_\ell$, which we denote by $x_1,\ldots,x_\ell$, is the colour of an edge $e_i$ for some $1\leq i\leq \ell$.  

Let $y_1,\ldots,y_\ell$ be the (pairwise distinct) colours of $wv_{\ell+1}',\ldots,wv_{2\ell}'$, respectively. By the same arguments as above, we find that each of those colours is also the colour of a pendant edge of $H(v)$ that is incident to a vertex $v_{\ell+i}$ for some $1 \leq i\leq \ell$. Note that $x_1,\ldots,x_\ell$, $y_1,\ldots,y_\ell$ are $2\ell$ pairwise distinct colours, as they are colours of edges incident to the same vertex, namely vertex~$w$. Hence, we can define a $k$-colouring~$c$ of $G$ by setting
$c(uv)=c'(u_hv_i)$ for every edge $uv\in E(G)$ with corresponding edge $u_hv_i\in E(G')$.\qed
\end{proof}

We note that the graph $G'$ in the proof of Lemma~\ref{l-claw} is not a line graph, as the gadget $H(v)$ is not a line graph: the vertices 
$v'_1,v'_2,v_1,w$
form a diamond and by adding the pendant edge incident to 
$v_1$
and the edge 
$wv'_{\ell+1}$
we obtain an induced subgraph of $H(v)$ that is not a line graph.
 
To handle the case where the forbidden induced subgraph~$H$ is a path, we make the following observation.

\begin{observation}\label{l-bounded}
If a graph $G$ of maximum degree~$k$ has a dominating set of size at most~$p$, then $G$ has at most $p(k+1)$ vertices.
\end{observation}

We use Observation~\ref{l-bounded} in the proof of the following lemma.

\begin{lemma}\label{l-path}
Let $k\geq 3$ and $t\geq 1$. Then $k$-{\sc Edge Colouring} is constant-time solvable for $P_t$-free graphs.
\end{lemma}

\begin{proof}
Let $G$ be a $P_t$-free graph with maximum degree~$\Delta$. 
If $\Delta\leq k-1$, then $G$ is $k$-edge colourable by Theorem~\ref{t-vizing}. If $\Delta\geq k+1$, then $G$ is not $k$-edge colourable.
Hence, we may assume that $\Delta=k$.
We may assume without loss of generality that $G$ is connected. We claim that $G$ has at most
$f(k,t)$ vertices for some function $f$ that only depends on $k$ and $t$. As we assume that $k$ and $t$ are constants, this means that
we can check in constant time if $G$ is $k$-edge colourable.

We prove the above claim by induction on $t$.
First suppose $t=4$ (and observe that if the claim holds for $t=4$, it also holds for $t\leq 3$). 
As $G$ is connected, $G$ has a dominating set of size~$2$ due to Lemma~\ref{l-p4}. 
Hence, by Observation~\ref{l-bounded}, $G$ has at most $f(k,2)=2(k+1)$ vertices. Now suppose $t\geq 5$. 
Let $X$ be an arbitrary minimum connected dominating set of~$G$. By Theorem~\ref{t-cs}, 
$G[X]$ is either $P_{t-2}$-free or isomorphic to $P_{t-2}$. 
In the first case we use the induction hypothesis to conclude that $G[X]$ has at most $f(k,t-2)$ vertices. Hence, $G$ has at most
$f(k,t-2)(k+1)$ vertices by Observation~\ref{l-bounded}.
In the second case, we find that $G$ has at most $(t-2)(k+1)$ vertices. We set $f(k,t)=\max\{f(k,t-2)(k+1),(t-2)(k+1)\}$.
\qed
\end{proof}

We are now ready to prove Theorem~\ref{t-main}, which we restate below.

\medskip
\noindent
{\bf Theorem~\ref{t-main}. (restated)}
{\it Let $k\geq 3$ be an integer and $H$ be a graph. If $H$ is a linear forest, then $k$-{\sc Edge Colouring} is polynomial-time solvable for  $H$-free graphs. Otherwise $k$-{\sc Edge Colouring} is \NP-complete even for $k$-regular $H$-free graphs.}

\begin{proof}
First suppose that $H$ contains a cycle $C_s$ for some $s\geq 3$. Then the class of $H$-free graphs is a superclass of the class of $C_s$-free graphs. This means that we can apply Theorem~\ref{t-ce1}.
From now on assume that $H$ contains no cycle, so $H$ is a forest.
Suppose that $H$ contains a vertex of degree at least~$3$. Then the class of $H$-free graphs is a superclass of the class of $K_{1,3}$-free graphs, which in turn forms a superclass of the class of line graphs. Hence, if $k$ is odd, then we apply Theorem~\ref{t-ce2}, and if $k$ is even, then we apply Lemma~\ref{l-claw}.
From now on assume that $H$ contains no cycle and no vertex of degree at least~$3$. Then $H$ is a linear forest, say with $\ell$ connected components. Let $t=\ell|V(H)|$. Then the class of $H$-free graphs is contained in the class of $P_t$-free graphs. Hence we may apply Lemma~\ref{l-path}. This completes the proof of Theorem~\ref{t-main}. \qed
\end{proof}

\section{Conclusions}\label{s-con}

We gave a complete complexity classification of $k$-{\sc Edge Colouring} for $H$-free graphs, showing a dichotomy between constant-time solvable cases and \NP-complete cases. We saw that this depends on~$H$ being a linear forest or not. It would be interesting to prove a dichotomy result for {\sc Edge Colouring} restricted to
$H$-free graphs. Note that due to Theorem~\ref{t-main} we only need to consider the case where $H$ is a linear forest. However, even determining the complexity for small linear forests~$H$, such as the cases where $H=2P_2$ and $H=P_4$, turns out to be a difficult problem. In fact, the computational complexity of {\sc Edge Colouring} for split graphs, or equivalently, $(2P_2,C_4,C_5)$-free graphs~\cite{FH77} and for $P_4$-free graphs has yet to be settled, despite the efforts towards solving the problem for these graph classes~\cite{CFK95,AMM15,LGZA15}.

On a side note, a graph is $k$-edge colourable if and only if its line graph is $k$-vertex colourable. In contrast to the situation for {\sc Edge Colouring}, the computational complexity of {\sc Vertex Colouring} has been fully classified for $H$-free graphs~\cite{KKTW01}. However, the complexity status of $k$-{\sc Vertex Colouring} is still open for $P_t$-free graphs if $k=3$ and $t\geq 8$ after Bonomo et al.~\cite{BCMSZ18} proved polynomial-time solvability for $k=3$ and $t=7$; see the survey~\cite{GJPS17} for further background information.

\section*{Acknowledgements}
The present work was done when the second author was visiting the University of Fribourg funded by a scholarship of the University of Fribourg.


\begin{thebibliography}{10}

\bibitem{Bo90}
H.~L. Bodlaender.
\newblock Polynomial algorithms for graph isomorphism and chromatic index on
  partial $k$-trees.
\newblock {\em Journal of Algorithms}, 11(4):631--643, 1990.

\bibitem{BCMSZ18}
F.~Bonomo, M.~Chudnovsky, P.~Maceli, O.~Schaudt, M.~Stein, and M.~Zhong.
\newblock Three-coloring and list three-coloring of graphs without induced
  paths on seven vertices.
\newblock {\em Combinatorica}, 38:779--801, 2018.

\bibitem{BLS99}
A.~Brandst\"adt, V.~B. Le, and J.~P. Spinrad.
\newblock {\em Graph Classes: A Survey}.
\newblock SIAM Monographs on Discrete Mathematics and Applications. Society for
  Industrial and Applied Mathematics (SIAM), 1999.

\bibitem{CE91}
L.~Cai and J.~A. Ellis.
\newblock {N}{P}-completeness of edge-colouring some restricted graphs.
\newblock {\em Discrete Applied Mathematics}, 30(1):15--27, 1991.

\bibitem{CS16}
E.~Camby and O.~Schaudt.
\newblock A new characterization of {$P_k$}-free graphs.
\newblock {\em Algorithmica}, 75(1), 2016.

\bibitem{CFK95}
B.-L. Chen, H.-L. Fu, and M.~T. Ko.
\newblock Total chromatic number and chromatic index of split graphs.
\newblock {\em Journal of Combinatorial Mathematics and Combinatorial
  Computing}, 17:137--146, 1995.

\bibitem{AMM15}
S.~M. de~Almeida, C.~P. de~Mello, and A.~Morgana.
\newblock Edge-coloring of split graphs.
\newblock {\em Ars Combinatoria}, 119:363--375, 2015.

\bibitem{FH77}
S.~F{\"o}ldes and P.~L. Hammer.
\newblock Split graphs.
\newblock {\em Congressus Numerantium}, XIX:311--315, 1977.

\bibitem{GJPS17}
P.~A. Golovach, M.~Johnson, D.~Paulusma, and J.~Song.
\newblock A survey on the computational complexity of colouring graphs with
  forbidden subgraphs.
\newblock {\em Journal of Graph Theory}, 84(4):331--363, 2017.

\bibitem{Ho81}
I.~Holyer.
\newblock The {N}{P}-completeness of edge-coloring.
\newblock {\em {SIAM} Journal on Computing}, 10(4):718--720, 1981.

\bibitem{Jo85}
D.~S. Johnson.
\newblock The {N}{P}-completeness column: An ongoing guide.
\newblock {\em J. Algorithms}, 6(3):434--451, 1985.

\bibitem{KL15}
S.~Kitaev and V.~V. Lozin.
\newblock {\em Words and Graphs}.
\newblock Monographs in Theoretical Computer Science. An EATCS Series.
  Springer, 2015.

\bibitem{KKTW01}
D.~Kr{\'a}l', J.~Kratochv\'{\i}l, {\relax Zs}.~Tuza, and G.~J. Woeginger.
\newblock Complexity of coloring graphs without forbidden induced subgraphs.
\newblock {\em Proc. WG 2001, LNCS}, 2204:254--262, 2001.

\bibitem{LG83}
D.~Leven and Z.~Galil.
\newblock {N}{P} completeness of finding the chromatic index of regular graphs.
\newblock {\em Journal of Algorithms}, 4(1):35--44, 1983.

\bibitem{LGZA15}
A.~R.~C. Lima, G.~Garcia, L.~Zatesko, and S.~M. de~Almeida.
\newblock On the chromatic index of cographs and join graphs.
\newblock {\em Eletronic Notes in Discrete Mathematics}, 50:433--438, 2015.

\bibitem{MFT13}
R.~C.~S. Machado, C.~M.~H. de~Figueiredo, and N.~Trotignon.
\newblock Edge-colouring and total-colouring chordless graphs.
\newblock {\em Discrete Mathematics}, 313(14):1547--1552, 2013.

\bibitem{Ma14}
D.~S. Malyshev.
\newblock The complexity of the edge 3-colorability problem for graphs without
  two induced fragments each on at most six vertices.
\newblock {\em Sib. elektr. matem. izv.}, 11:811--822, 2014.

\bibitem{Ma17}
D.~S. Malyshev.
\newblock Complexity classification of the edge coloring problem for a family
  of graph classes.
\newblock {\em Discrete Mathematics and Applications}, 27:97--101, 2017.

\bibitem{OMS98}
C.~Ortiz, N.~Maculan, and J.~L. Szwarcfiter.
\newblock Characterizing and edge-colouring split-indifference graphs.
\newblock {\em Discrete Applied Mathematics}, 82(1-3):209--217, 1998.

\bibitem{Vi64}
V.~G. Vizing.
\newblock On an estimate of the chromatic class of a $p$-graph.
\newblock {\em Diskret. Analiz.}, 3:25--30, 1964.

\end{thebibliography}
\end{document}